\documentclass[a4paper,12pt]{article}
\usepackage{amsmath,amssymb,amsfonts,amsthm}
\newtheorem{theorem}{Theorem}[section]
\newtheorem{lemma}[theorem]{Lemma}

\newtheorem{corollary}[theorem]{Corollary}

\newcommand{\den}{\varepsilon}

\newcommand{\Rdm}{\mathbb{R}^2_{+}}

\newcommand{\Ld}{\Delta}
 \newcommand{\Ldt}{{^{(3)}\Delta}}
\newcommand{\Ldta}{{^{(7)}}\Delta}

\newcommand{\po}{d}

\newcommand{\mt}{\bar{m}_2}

\title{Bounds for axially symmetric linear perturbations for the  extreme Kerr black hole}
\author{Sergio Dain and Ivan Gentile de Austria\\
 Facultad de Matem\'atica, Astronom\'{i}a y F\'{i}sica, FaMAF,\\
  Universidad Nacional de C\'ordoba,\\
  Instituto de F\'{\i}sica Enrique Gaviola, IFEG, CONICET,\\
  Ciudad Universitaria, (5000) C\'ordoba, Argentina.}

\begin{document}

\maketitle

\begin{abstract}
  We obtain remarkably simple integral bounds for axially symmetric linear
  perturbations for the extreme Kerr black hole in terms of conserved energies.
  From these estimates we deduce pointwise bounds for the perturbations outside
  the horizon.
\end{abstract}

\section{Introduction}
In this article we continue the work initiated in \cite{Dain:2014iba} destined
to study the linear stability of extreme Kerr black hole under axially
symmetric gravitational perturbations using conserved energies. For a general
introduction to the subject and a list of relevant references we refer to
\cite{Dain:2014iba}. The main result of that article is that there exists a
positive definite and conserved energy for the axially symmetric gravitational
perturbations. In the present article, using this energy, we prove the
existence of integral bounds for the first and second derivatives of the
perturbation. In particular, these bounds imply pointwise estimates for the
perturbation outside the black hole horizon. This result is presented in
theorem \ref{theorem1} and corollary \ref{corolario}. In the following we
introduce the basic definitions and notation needed to formulate the theorem
and then we discuss the meaning and scope of these estimates.

Axially symmetric perturbations are characterized by two functions $\sigma_1$
and $\omega_1$ which represent the linear perturbation of the norm and the
twist of the axial Killing vector. The coordinates system $(t,\rho,z)$ is fixed
by the maximal-isothermal gauge condition.  Partial derivatives with respect to
the space coordinates $(\rho,z)$ is denoted by $\partial_\rho$ and $\partial_z$
respectively and partial derivative with respect to $t$ is denotes with a
dot. We use the following notation to abbreviate the products of gradients in the spatial
coordinates $(\rho,z)$ for functions $f$ and $g$
\begin{equation}
  \label{eq:18}
 \partial f \partial g = \partial_\rho f \partial_\rho g+  \partial_z
 f \partial_z g, \quad|\partial f|^2=(\partial_\rho f)^2+(\partial_z f)^2.
\end{equation}
 The 2-dimensional Laplacian $\Ld$ is defined by 
\begin{equation}
  \label{eq:delta}
\Ld f = \partial^2_\rho f +  \partial^2_z  f,
\end{equation}
and  the operators  $\Ldt$  and $\Ldta$ are  defined by
\begin{equation}
  \label{eq:delta3}
  \Ldt f=\Ld f+\frac{\partial_\rho f}{\rho}, \quad 
 \Ldta f=\Ld f+5\dfrac{\partial_{\rho} f}{\rho}.
\end{equation}
The operators $\Ldt$ and $\Ldta$ correspond to the flat Laplace operator in
3-dimensions and 7-dimensions respectively written in cylindrical coordinates
and acting on axially symmetric functions.

The domain for the space coordinates $(\rho,z)$ is the half plane $\Rdm$
defined by $0\leq \rho<\infty$, $-\infty < z < \infty$.  The axis of symmetry
is give by $\rho=0$.  In these coordinates, the horizon is located at the
origin $r=0$, where $r=\sqrt{\rho^2 + z^2}$. We follow the same notation and
conventions used in \cite{Dain:2014iba} and we refer to that article for
further details.

The linear equations for axially symmetric gravitational perturbations for the
extreme Kerr black hole in the maximal-isothermal gauge were obtained in
\cite{Dain:2014iba}. In appendix \ref{linearized-equations} we briefly review
the set of equations needed in the proof of theorem \ref{theorem1}. The
background quantities are denoted with a subindex $0$, and the first order
perturbation with a subindex $1$.  The square norm of the axial Killing vector
of the background extreme Kerr metric is denoted by $\eta_0$ and the background
function $\sigma_0$ is defined by
\begin{equation}
  \label{eq:app35}
e^{\sigma_0}=\frac{\eta_0}{\rho^2}.
\end{equation}
The other relevant background quantities are the twist $\omega_0$ and the
function $q_0$. In appendix \ref{s:falloff} we review the behaviour of these
explicit functions.

It is useful to define the following rescaling of $\omega_1$
\begin{equation}
 \label{eq:omegabarra}
 \bar{\omega}_{1}=\dfrac{\omega_{1}}{\eta_{0}^{2}}.
\end{equation}
The extreme Kerr solution depends on only one parameter $m_0$ which represents
the total mass of the black hole. The first order linearization of the total
mass of the spacetime $m_1$ vanished. The second order expansion of the total
mass $m_2$ provides a positive definite and conserved quantity for the
perturbation (see \cite{Dain:2014iba}) that is given explicitly in equation
(\ref{eq:3mass}). Taking time derivatives of the linear equations we get an
infinity number of conserved quantities that have the same form as $m_2$ but in
terms of the time derivatives of the corresponding quantities. For the result
presented bellow, we will make use of $\mt$ which is obtained taking one time
derivative of the equations. The explicit expression for $\mt$ is given in
equation (\ref{eq:3mass2}). The conserved quantities $m_2$ and $\mt$ depend
only on the initial conditions for the perturbation.

We will assume in the following that the perturbations satisfy the fall off and
boundary conditions discussed in detail in \cite{Dain:2014iba}. Physically,
these conditions imply that the system is isolated (and hence it has finite
total energy $m_2$) and that that the perturbations do not change the angular
momentum of the background (i.e. $\omega_1$ vanished at the axis).

\begin{theorem}
 \label{theorem1}
 Axially symmetric linear gravitational perturbations
 $(\sigma_{1},\, \omega_{1})$ for the extreme Kerr black hole satisfy the
 following bound
\begin{align}
\label{cota1}
  & \int_{ \Rdm}\left(\dfrac{1}{2}\eta_{0}^{2}|\partial\bar{\omega}_{1}|^{2}+|\partial\eta_{0}|^{2}\bar{\omega}_{1}^{2}+  |\partial\sigma_{1}|^{2}+\dfrac{\sigma_{1}^{2}}{r^{2}}\right)\rho d\rho dz\leq  C m_{2} , \\
\label{cota2}
  & \int_{\Rdm}\left( \left(\Ldt\sigma_{1}\right)^{2}+\eta_{0}^{2}\left(\Ldta\bar{\omega}_{1}\right)^{2} \right)e^{-2(q_{0}+\sigma_{0})} \rho d\rho dz\leq  C \left(\mt+m_{2} \right),
 \end{align}
 where $C$ is a positive constant that depends only on the mass $m_0$ of the
 background extreme Kerr black hole.
\end{theorem}
The conserved quantity $m_2$ involves first spatial derivatives of
$\sigma_1$ and $\omega_1$, the quantity $\mt$ involves also second spatial
derivatives of $\sigma_1$ and $\omega_1$. Note however, that these terms
appears in a rather complicated way and hence it is by no means obvious that
$m_2$ and $\mt$ satisfy the bound (\ref{cota1}) and (\ref{cota2}).

Besides $\sigma_1$ and $\omega_1$, gravitational perturbations involve other
quantities (the shift vector $\beta_1$, the metric function $q_1$, the second
fundamental form $\chi^{AB}_1$, see \cite{Dain:2014iba} for the details). These
other functions are, in principle, calculated in terms $\sigma_1$ and
$\omega_1$ using the coupled system of equations. It is remarkable that the
estimates (\ref{cota1}) and (\ref{cota2}) can be written purely in terms of the
geometric functions $\sigma_1$ and $\omega_1$ (which precisely encode the
dynamical degree of freedom of the system) without involving the other
functions.

The functions $(\sigma_1, \omega_1)$ satisfy the linear evolution equations
(\ref{eq:evolucion-sigma-kerr})--(\ref{eq:evolucion-omega-kerr}). These
equations have the well known structure of a wave map coupled with a non-trivial
background metric. Recently, a model problem for an analogous wave map (but
without the coupling) was studied in \cite{ionescu14}. Remarkably enough the
estimates proved in theorem \ref{theorem1} make use only on the wave map
structure of the equations but they hold for the complete coupled system.
Also, these estimates are robust in the sense that they make use of energies
that are also available for the  non-linear equations.

Since the estimates (\ref{cota1}) and (\ref{cota2}) essentially control up to
second derivatives of the functions $(\sigma_1, \omega_1)$, using Sobolev
embeddings we can obtain pointwise bounds.  This is, of course, one of the main
motivations to obtain these kind of estimates. Note however, that the different
terms are multiplied by the background functions. Particularly relevant is the
factor $e^{-2(\sigma_0+q_0)}$ that appears in (\ref{cota2}). This explicit
  function is positive, goes to $1$ at infinity but vanished like $r^2$ at the
  origin, where the black hole horizon is located. That means that the estimate
  (\ref{cota2}) degenerate at the origin (but not at infinity) and we can not
  expect to control pointwise the functions $(\sigma_1, \omega_1)$ at the
  origin using only this estimate. In the following corollary we prove
  pointwise bounds outside the horizon. 

Let $\delta >0$ an arbitrary, small, number. We define the following two domains
\begin{align}
\Omega_\delta &=\left\lbrace (\rho,z)\in\Rdm, \text{such that}\,\, 0<\delta \leq r \right\rbrace \\
\Gamma_\delta &=\left\lbrace  (\rho,z)\in\Rdm, \text{such that}\,\, 0<\delta \leq\rho \right\rbrace
\end{align}

We have the following result. 
\begin{corollary}
 \label{corolario}
 Under the same assumptions of theorem \ref{theorem1}, the following pointwise bounds hold
   \begin{align}
 \label{cota3}
 \underset{\Omega_{\delta}}{\sup}|\sigma_{1}|\leq   C_\delta  \left(\mt+m_{2} \right), \\
 \label{cota4}
 \underset{\Gamma_{\delta}}{\sup}|\bar{\omega}_{1}|\leq  C_\delta \left(\mt+m_{2} \right), 
 \end{align}
where the constant $C_\delta$ depends on $m_0$ and $\delta$. 
\end{corollary}
The bounds (\ref{cota3}) and (\ref{cota4}) are not intended to be sharp, they
are meant as example of possible pointwise bounds that can be deduced from
(\ref{cota1}) and (\ref{cota2}).  It is certainly conceivable that sharped
weighted bounds can be proved using the estimates (\ref{cota1}) and
(\ref{cota2}). But it is also clear that no pointwise bound at the horizon can
be proved using these estimates, because the factor $e^{-2(\sigma_0+q_0)}$
vanishes there. The situation strongly resemble the problem studied in
\cite{Dain:2012qw}. In that article the wave equation on the extreme
Reissner-Nordstr\"om black hole was analyzed using conserved energies.  In
order to prove a pointwise bound at the horizon it was not enough with the
first two energies. An extra energy which involves ``integration in time'' was
needed. It is a relevant open question whether the same strategy can be applied
to the present case, which is certainly much more complicated.

\section{Proof of Theorem \ref{theorem1}}
\label{sec:proof-theor}
In this section we prove theorem \ref{theorem1}. We begin with the estimate
(\ref{cota1}), which represents the most important part of the theorem.  In the
integrand on the left hand side of (\ref{cota1}) the terms involve up to first
derivatives of $\sigma_1$ and $\omega_1$. The integral is bounded only with the
energy $m_2$, the higher order energy $\mt$ is not needed for this estimate.
Moreover, to prove the bound (\ref{cota1}) we will make use only of the last
three terms in the energy density $\den_2$ given by (\ref{eq:2epsilon}). Note
that in $\den_2$  appears the same terms as in the integrand on the left hand
side of (\ref{cota1}). However they appear arranged in different form
(i.e. there are many cross products) and it not obvious how to deduce
the bound (\ref{cota1}). 

The proof of (\ref{cota1}) can be divided in two parts. The first part consists
in integral estimates, this is the subtle part of the proof. The second part
consists on pointwise estimates. In the arguments, we make repeatedly use of
different forms of the standard Cauchy inequality, for readability we summarize
them bellow. Let $a_1\cdots a_n$ be arbitrary real numbers, then we have
\begin{equation}
 \label{eq:desigualdaddecauchy}
 a_1^{2}+a_2^{2}\geq\dfrac{1}{2}(a_1+a_2)^{2},
\end{equation}
and, in general, 
\begin{equation}
 \label{eq:desigualdaddecauchy1}
a_{1}^{2}+a_{2}^{2}+...\,+a_{n}^{2}\geq\dfrac{1}{n}\left(a_{1}+a_{2}+...\,+a_{n}\right)^{2}.
\end{equation}
Let $\lambda>0$, then 
\begin{equation}
 \label{eq:desigualdadgeneraldecauchy}
 a_1a_2\leq\lambda a_1^{2}+\dfrac{a_2^{2}}{4\lambda}.
\end{equation}

In the following two lemmas we prove the relevant integral estimates. The
relevance of lemma \ref{l:1} in the proof of the estimate (\ref{cota1}) is
clear: in this lemma the bound for the fourth term in (\ref{cota1}) is proved.
This integral bound is the key to prove the bounds for the second and the third
term in (\ref{cota1}). We will see in the following, that in order to prove
these bounds we will need the integral estimate proved in lemma \ref{l:2} with
$v=\bar\omega_1$.

\begin{lemma}
\label{l:1}
Consider the mass $m_2$ given by (\ref{eq:3mass}) and
(\ref{eq:2epsilon}). Then, the following inequality holds
\begin{equation}
 \label{eq:8-desigualdad8}
 m_{2}\geq\int_{\Rdm}\dfrac{\sigma_{1}^{2}}{r^{2}} \, \rho d\rho dz.
\end{equation}

\end{lemma}

\begin{proof}
  The mass $m_2$ is the second variation of the total ADM mass (see
  \cite{Dain:2014iba}). The first three terms in (\ref{eq:2epsilon}) correspond
  to the dynamical part of the mass (these terms vanished for stationary
  solutions), the last three terms correspond to the stationary part of the
  mass. These terms are precisely the second variation of the mass functional
  extensively studied in connection with the mass angular momentum inequality
  (see \cite{Dain06c}, \cite{dain12} and reference therein). In a recent
  article \cite{Schoen:2012nh}, an important estimate has been proved for the
  second variation of this functional in terms of the distance function in the
  hyperbolic plane. From lemma 2.3 in \cite{Schoen:2012nh} we deduce the
  following inequality
\begin{equation} 
\label{eq:5desigualdad5}
m_2\geq2\int_{\Rdm}|\partial
\mathbf{d}((\eta_{1},\omega_{1}),(\eta_{0},\omega_{0}))|^{2} \, \rho d\rho dz,
 \end{equation} 
 where $\mathbf{d}$ is the distance function in the hyperbolic plane between the two
 points $(\eta_{1},\omega_{1})$ and $(\eta_{0},\omega_{0})$, where
 $\eta_1=\rho^2 e^{\sigma_1}$ (see, for example, \cite{Dain06c} for the
 explicit expression of $\mathbf{d}$).

To obtain the desired  lower bound for  the right hand side of the inequality
(\ref{eq:5desigualdad5}) we first use the following weighted Poincare
inequality proved in \cite{Dain05d} (equation (31) in \cite{Dain05d} with $\delta=-1/2$)
\begin{equation}
 \label{eq:6desigualdad6}
  2\int_{\Rdm}|\partial \mathbf{d}|^{2} \, \rho d\rho dz\geq\int_{\Rdm} \frac{\mathbf{d}^{2}}{r^{2}}\,\rho d\rho dz,
\end{equation}
and then we use the following bound for the distance function $\mathbf{d}$
proved in \cite{Dain06c} (see equation (138) in that reference)
\begin{equation}
 \label{eq:7desigualdad7}
 |\mathbf{d}|\geq|\sigma_{1}|.
\end{equation}

\end{proof}

\begin{lemma}
\label{l:2}
Let $\eta_0$ and  $\omega_0$ be the norm and the twist function for the extreme
Kerr black hole, and let $v$ be an arbitrary smooth function with
compact support outside  the axis. Then, the following inequality holds
\begin{equation}
 \label{eq:oes}
 \int_{\Rdm}|\partial\omega_{0}|^{2} v^{2}\rho d\rho dz \leq
 3\int_{\Rdm}|\partial\eta_{0}|^{2} v^{2}\rho d\rho
 dz+\int_{\Rdm}\eta_{0}^{2}|\partial v |^{2}\rho d\rho dz.   
\end{equation}  
\end{lemma}

\begin{proof}
 We will use that the background function $\eta_0$ and $\omega_0$ satisfy
 equation (\ref{eq:7eta}).
Let $v$ is an arbitrary function with  compact support outside of the axis. We
multiply (\ref{eq:7eta}) by $\eta_{0}^{-2\delta}v^{2}$ (where $\delta$ is an
arbitrary number) and integrate, we obtain
\begin{equation}
 \label{eq:9eta}
 \int_{\Rdm}\eta_{0}^{-2\delta}v^{2}\Ldt (\ln\eta_{0})\rho d\rho dz=-\int_{\Rdm}\dfrac{|\partial\omega_{0}|^{2}}{\eta_{0}^{2}}\eta_{0}^{-2\delta}v^{2}\rho d\rho dz.
\end{equation}
Integrating by parts the left hand side  of (\ref{eq:9eta}) we obtain the
following useful identity 
\begin{multline}
 \label{eq:10eta}
 \int_{\Rdm}|\partial\omega_{0}|^{2}\eta_{0}^{-2\delta-2}v^{2}\rho d\rho dz=-2\delta\int_{\Rdm}\eta_{0}^{-2\delta-2}|\partial\eta_{0}|^{2}v^{2}\rho d\rho dz
 +\\
+2\int_{\Rdm}\eta_{0}^{-2\delta-1}v\partial v\partial\eta_{0}\rho d\rho dz
\end{multline}  
We take  $\delta=-1$ in (\ref{eq:10eta}), we obtain  
\begin{align}
 \label{eq:desigualdadomegabarra5}
  \int_{\Rdm}|\partial\omega_{0}|^{2}v^{2}\rho d\rho dz
  &=2\int_{\Rdm}|\partial\eta_{0}|^{2} v^{2}\rho d\rho dz 
    +2\int_{\Rdm}\eta_{0}v\partial v\partial\eta_{0}\rho d\rho dz,\\
  &\leq 3\int_{\Rdm}|\partial\eta_{0}|^{2} v^{2}\rho d\rho
    dz+\int_{\Rdm}\eta_{0}^{2} |\partial v|^{2}\rho d\rho
    dz,  
\label{eq:desigualdadomegabarra5a}
\end{align}  
where, to obtain line  (\ref{eq:desigualdadomegabarra5a}) we have used in the
second term of the right hand side of (\ref{eq:desigualdadomegabarra5}) the
inequality  (\ref{eq:desigualdadgeneraldecauchy}) with 
$a_1=\eta_{0}\partial\bar{\omega}_{1}$, 
$a_2=\partial\eta_{0}\bar{\omega}_{1}$ and $\lambda=1/2$.  
\end{proof}

We prove now the pointwise bounds in terms of the energy density $\den_{2}$. In
the following we denote by $C$ a generic positive constant that depends only on
the background parameter $m_0$.  

We begin with the first term in the integrand in (\ref{cota1}).  From the explicit expression
for $\den_2$ given in (\ref{eq:2epsilon}), keeping only the fifth and sixth
terms, we obtain
\begin{align}
 \label{eq:desigualdadomegabarra}
\frac{\den_{2}}{\rho}
  &\geq\left(\partial\left(\omega_{1}\eta_{0}^{-1}\right)-\eta_{0}^{-1}\sigma_{1}\partial\omega_{0}\right)^{2}+\left(\eta_{0}^{-1}\sigma_{1}\partial\omega_{0}-\omega_{1}\eta_{0}^{-2}\partial\eta_{0}\right)^{2}\\
  &=(\eta_{0}\partial\bar{\omega}_{1}+\bar{\omega}_{1}\partial\eta_{0}-\eta_{0}^{-1}\sigma_{1}\partial\omega_{0})^{2}+(\eta_{0}^{-1}\sigma_{1}\partial\omega_{0}-\bar{\omega}_{1}\partial\eta_{0})^{2} 
\label{eq:desigualdadomegabarraa}
\end{align}
where in (\ref{eq:desigualdadomegabarraa}) we have used the definition of
$\bar{\omega}_{1}$ given in (\ref{eq:omegabarra}).  We use the Cauchy
inequality (\ref{eq:desigualdaddecauchy}) in \eqref{eq:desigualdadomegabarraa}
to finally obtain
\begin{equation}
  \label{eq:1b}
 \frac{\den_2}{\rho} \geq\dfrac{1}{2}\eta_{0}^{2}(\partial\bar{\omega}_{1})^{2}.
\end{equation}

From the second term in (\ref{cota1}) we take $\den_2$ given in
(\ref{eq:2epsilon}) and keep only the last term, we obtain
\begin{align}
 \label{eq:desigualdadomegabarra3} 
\frac{\den_2}{\rho} &\geq\left(\dfrac{\partial\omega_{0}}{\eta_{0}}\sigma_{1}-\partial\eta_{0}\bar{\omega}_{1}\right)^{2},\\
&=\dfrac{|\partial\omega_{0}|^2}{\eta^2_{0}}\sigma^2_{1}+|\partial\eta_{0}|^2\bar{\omega}^2_{1}-2
  \sigma_{1} \bar{\omega}_{1} \dfrac{\partial\omega_{0}}{\eta_{0}}\partial\eta_{0},\label{eq:desigualdadomegabarra3a} \\
 &\geq-\dfrac{|\partial\omega_{0}|^{2}}{\eta_{0}^{2}}\sigma_{1}^{2}+\dfrac{1}{2}|\partial\eta_{0}|^{2}\bar{\omega}_{1}^{2}, 
\label{eq:desigualdadomegabarra3b}\\
 &\geq-\dfrac{C}{r^{2}}\sigma_{1}^{2}+\dfrac{1}{2}|\partial\eta_{0}|^{2}\bar{\omega}_{1}^{2} ,
\label{eq:desigualdadomegabarra3c}
\end{align}
where in the inequality (\ref{eq:desigualdadomegabarra3b})  we have used the
Cauchy inequality (\ref{eq:desigualdadgeneraldecauchy}) with $\lambda=1$ and in
 line  (\ref{eq:desigualdadomegabarra3c}) we have used the bound
(\ref{eq:falloffso}) for the background quantities. We have obtained 
\begin{equation}
 \label{eq:desigualdadomegabarra3d}
\frac{\den_2}{\rho} +\dfrac{C}{r^{2}}\sigma_{1}^{2}\geq\dfrac{1}{2}|\partial\eta_{0}|^{2}\bar{\omega}_{1}^{2}. 
\end{equation}
We integrate the bounds  \eqref{eq:1b} and \eqref{eq:desigualdadomegabarra3d},
and use the integral bound (\ref{eq:8-desigualdad8}) to obtain
\begin{equation}
  \label{eq:19}
 \int_{
   \Rdm}\left(\dfrac{1}{2}\eta_{0}^{2}|\partial\bar{\omega}_{1}|^{2}+|\partial\eta_{0}|^{2}\bar{\omega}_{1}^{2} 
+\dfrac{\sigma_{1}^{2}}{r^{2}}\right)\rho d\rho dz\leq  C m_{2} . 
\end{equation}
To prove \eqref{cota1} it only remains to bound the term
$ |\partial\sigma_{1}|^{2}$. For that term,  we use the  fourth term in
(\ref{eq:2epsilon}) to obtain 
\begin{align}
 \label{eq:desigualdadderivadasigma}
 \frac{\den_2}{\rho} &\geq\left(\partial\sigma_{1}+\omega_{1}\eta_{0}^{-2}\partial\omega_{0}\right)^{2}\\
 &=(\partial\sigma_{1}+\bar{\omega}_{1}\partial\omega_{0})^{2},\label{eq:desigualdadderivadasigmaa}\\
 &=|\partial\sigma_{1}|^{2}+\bar{\omega}_{1}^{2}|\partial\omega_{0}|^{2}
   +2 \bar{\omega}_{1}\partial\sigma_{1} \partial\omega_{0},\label{eq:desigualdadderivadasigmab}\\
 &\geq\dfrac{1}{2}|\partial\sigma_{1}|^{2}-|\partial\omega_{0}|^{2}\bar{\omega}_{1}^{2},
\label{eq:desigualdadderivadasigmac}
\end{align}
where in line \eqref{eq:desigualdadderivadasigmaa} we have just used the
definition of $\bar{\omega}_{1}$ and in line
\eqref{eq:desigualdadderivadasigmac} we have used inequality
(\ref{eq:desigualdadgeneraldecauchy}) with $\lambda=\dfrac{1}{4}$. Then, we
have obtained
\begin{equation}
 \label{eq:desigualdadderivadasigmab1} 
\frac{\den_{2}}{\rho}+|\partial\omega_{0}|^{2}\bar{\omega}_{1}^{2}\geq\dfrac{1}{2}|\partial\sigma_{1}|^{2}.
\end{equation}
We integrate the pointwise estimate (\ref{eq:desigualdadderivadasigmab1}), to
handle second term on the left hand side of
(\ref{eq:desigualdadderivadasigmab1}) we use the integral bound
(\ref{eq:8-desigualdad8}) $v=\bar\omega_1$ and the bound \eqref{eq:19}. Hence
we have obtained the desired estimate \eqref{cota1}.

We turn to  the bound (\ref{cota2}) which involves second derivatives of the
functions $\sigma_1$ and $\bar \omega_1$ and hence we need the higher order
mass $\mt$.

We begin with the term with $\sigma_1$.  We use the evolution equation
(\ref{eq:evolucion-sigma-kerr}) to obtain
\begin{align}
 \label{eq:sdsigma}
\left(\Ldt\sigma_{1}\right)^{2}&
                                 =\left(e^{2(\sigma_{0}+q_{0})}\dot{p}+\dfrac{2}{\eta^2_0}\left(\sigma_{1}|\partial\omega_{0}|^{2}-\partial \omega_{1}\partial\omega_{0} \right)\right)^{2}\\
 &=\left(e^{2(\sigma_{0}+q_{0})}\dot{p}+\dfrac{2}{\eta^2_0}\left(\sigma_{1}|\partial\omega_{0}|^{2}-2\eta_{0} \bar{\omega}_{1}\partial\eta_{0}\partial\omega_{0} 
 +\eta_{0}^{2} \partial\omega_{0}\partial \bar{\omega}_{1}\right)\right)^{2}\label{eq:sdsigma2}\\
 &\leq 4e^{4(\sigma_{0}+q_{0})}\dot{p}^{2}+\dfrac{16}{\eta_{0}^{4}}\sigma_{1}^{2}|\partial\omega_{0}|^{4}
+\dfrac{64}{\eta_{0}^{2}}\left(\partial\eta_{0}\partial\omega_{0}\right)^{2}\bar{\omega}_{1}^{2}
 +16\left(\partial\omega_{0}\partial \bar{\omega}_{1}\right)^{2} \label{eq:sdsigma3}\\
&\leq 4e^{4(\sigma_{0}+q_{0})}\dot{p}^{2}+C \dfrac{|\partial\omega_{0}|^{2}}{\eta_{0}^{2}}\dfrac{\sigma_{1}^{2}}{r^{2}} 
 +\dfrac{64}{\eta_{0}^{2}}|\partial\eta_{0}|^{2}|\partial\omega_{0}|^{2}\bar{\omega}_{1}^{2}
 +16|\partial\bar{\omega}_{1}|^{2}|\partial\omega_{0}|^{2}, \label{eq:sdsigma4}
\end{align}
where in line (\ref{eq:sdsigma2}) we have used the definition of $\bar{\omega}_{1}$ (\ref{eq:omegabarra}),  in line \eqref{eq:sdsigma3} we have used the inequality  (\ref{eq:desigualdaddecauchy1}) and line \eqref{eq:sdsigma4} follows from the bound  (\ref{eq:falloffso}) and the Cauchy-Schwartz inequality.

Multiplying by  $e^{-2(\sigma_{0}+q_{0})}$ the inequality  \eqref{eq:sdsigma4} we obtain
\begin{align}
\label{eq:sdsigma5}
e^{-2(\sigma_{0}+q_{0})}\left(\Ldt\sigma_{1}\right)^{2} &\leq 4e^{2(\sigma_{0}+q_{0})}\dot{p}^{2}+
e^{-2(\sigma_{0}+q_{0})}\dfrac{|\partial\omega_{0}|^{2}}{\eta_{0}^{2}} \left(
C \dfrac{\sigma_{1}^{2}}{r^{2}}
+64|\partial\eta_{0}|^{2}\bar{\omega}_{1}^{2}
 +16|\partial\bar{\omega}_{1}|^{2} \eta_0^2 \right)\\
 &\leq 2\frac{\bar{\den}_{2}}{\rho}+C\left(   \dfrac{\sigma_{1}^{2}}{r^{2}}
 +|\partial\eta_{0}|^{2} \bar{\omega}_{1}^{2}
 +|\partial\bar{\omega}_{1}|^{2} \eta_0^2\right) \label{eq:sdsigma6}
\end{align}
where in line \eqref{eq:sdsigma6}  we have used (\ref{eq:2epsilon2}) and   the bound  \eqref{eq:combi}.

Integrating (\ref{eq:sdsigma6}) and using the previous bounds we finally obtain
\begin{eqnarray}
 \label{eq:segundaderivadasigma2}
 \int_{\Rdm}e^{-2(\sigma_{0}+q_{0})}\left(\Ldt\sigma_{1}\right)^{2}\rho d\rho dz\leq C\left(\mt+ m_{2}\right).
\end{eqnarray}

To estimate the second derivatives of $\bar{\omega}_{1}$ we proceed in a
similar way. We will make use of the evolution equation
(\ref{eq:evolucion-omega-kerr}). First,  it is  useful to write this equation in terms
of $\bar\omega_1$ instead of $\omega_1$. To do that we first obtain the
following relation 
\begin{equation}
  \label{eq:22}
  \eta_0^2  \, \Ldta \bar\omega_1 = \Ldt \omega_1 -\frac{4}{\rho} \partial_\rho
  \omega_1-4\partial \omega_1 \partial \sigma_0 -2\omega_1 \Ldt
  \sigma_0+8\frac{\omega_1}{\rho} \partial_\rho \sigma_0+4 \omega_1 |\partial \sigma_0|^2,
\end{equation}
where we have used the definition of $\bar \omega_1$ given in
(\ref{eq:omegabarra}), the expression of $\eta_0$ in terms of $\sigma_0$ given
in (\ref{eq:app35}), the definitions of the operators $\Ldt$ and $\Ldta$ given
in (\ref{eq:delta3}) and the
identity (\ref{eq:21}).  Using the evolution equation
(\ref{eq:evolucion-omega-kerr}) and equation (\ref{eq:22}) we obtain
\begin{equation}
  \label{eq:23}
  \eta_0^2  \, \Ldta \bar\omega_1 = e^{2(\sigma_{0}+q_{0}}\dot{d}-2\eta_{0}^{2}\bar{\omega}_{1} \Ldt\sigma_{0}
 -2\eta_{0}^{2}\partial\sigma_{0}\partial\bar{\omega}_{1}  + 2\partial\omega_{0}\partial\sigma_{1}.
\end{equation}
To obtain the estimate, we take the square of each side of equation
(\ref{eq:23}) and use the Cauchy inequality (\ref{eq:desigualdaddecauchy1}) to obtain
\begin{align}
 \eta_{0}^{4}\left(\Ldta\bar{\omega}_{1}\right)^{2}  &\leq 4e^{4(\sigma_{0}+q_{0})}\dot{d}^{2}+16\eta_{0}^{4}\left(\Ldt\sigma_{0}\right)^{2}\bar{\omega}_{1}^{2} 
 +16\eta_{0}^{4}\left(\partial\sigma_{0}\partial\bar{\omega}_{1}\right)^{2}
 +16\left(\partial\omega_{0}\partial\sigma_{1}\right)^{2}  \label{eq:sdomega1}\\
 &\leq 4e^{4(\sigma_{0}+q_{0})}\dot{d}^{2}+16|\partial \omega_0|^2 \left( |\partial \omega_0|^2 \bar{\omega}_{1}^{2} +  |\partial\sigma_{1}|^{2}\right)+
 16\eta_{0}^{4}|\partial\sigma_{0}|^{2}|\partial\bar{\omega}_{1}|^{2} \label{eq:sdomega2}
\end{align}
where in line \eqref{eq:sdomega2} we have used the  Cauchy-Schwartz inequality
and equation (\ref{eq:sta-sigma}) to substitute the factor $\Ldt\sigma_{0}$.
We multiply by $e^{-2(\sigma_{0}+q_{0})}\eta_{0}^{-2}$  each side of   inequality \eqref{eq:sdomega2} 
\begin{multline}
\label{eq:interm}
  \eta_{0}^{2}e^{-2(\sigma_{0}+q_{0})}\left(\Ldta\bar{\omega}_{1}\right)^{2} 
\leq 4\dfrac{e^{2(\sigma_{0}+q_{0})}}{\eta_{0}^{2}}\dot{d}^{2} 
+ 16e^{-2(\sigma_{0}+q_{0})} \frac{|\partial \omega_0|^2}{\eta^2_0} 
\left( |\partial\omega_0|^2 \bar{\omega}_{1}^{2}+ |\partial\sigma_{1}|^{2}\right) \\
+ 16e^{-2(\sigma_{0}+q_{0})}|\partial\sigma_{0}|^{2}\eta_{0}^{2} |\partial\bar{\omega}_{1}|^{2}              
\end{multline}
Then, we bound the first term on the right hand side of inequality
\eqref{eq:interm} with the energy density $\bar{\den}_2$(\ref{eq:2epsilon2}), for the other
terms we use the inequalities (\ref{eq:combi}) and (\ref{eq:combi2})  to bound the background functions by a
constant $C$. We  obtain
\begin{equation}
  \label{eq:24}
  \eta_{0}^{2}e^{-2(\sigma_{0}+q_{0})}\left(\Ldta\bar{\omega}_{1}\right)^{2}
  \leq 2\frac{\bar \den_2}{\rho}+ C \left( 
 |\partial\omega_0|^2 \bar{\omega}_{1}^{2}+ |\partial\sigma_{1}|^{2} 
+ \eta_{0}^{2} |\partial\bar{\omega}_{1}|^{2} \right).
\end{equation}

Integrating \eqref{eq:24} and using (\ref{cota1}) we finally have
\begin{equation}
 \label{eq:segundaderivadaomega2}
 \int_{\Rdm}\eta_{0}^{2}e^{-2(\sigma_{0+}q_{0})}\left(\Ldta\bar{\omega}_{1}\right)^{2}\rho d\rho dz\leq C\left( \mt + m_{2}\right).
\end{equation}

\section{Proof of  Corollary \ref{corolario}}
\label{sec:proof-corollary}
In the proof of the corollary \ref{corolario} we essentially use an
appropriated variant of the Sobolev embedding and standard cut off functions
arguments.

Let $\chi:\mathbb{R}\to \mathbb{R}$
be a smooth cut off function such that $\chi \in C^\infty(\mathbb{R})$,
$0\leq\chi\leq 1$, $\chi(r)=1$ for $0\leq r \leq 1$, $\chi(r)=0$ for $2\leq
r$. Define $\chi_\delta(r) = \chi(r/\delta)$.

Consider the following function
\begin{equation}
  \label{eq:2}
  \bar\sigma_1=(1-\chi_\delta)\sigma_1
\end{equation}
Note that $ \bar\sigma_1=0$ in $B_\delta$ and $\bar \sigma_1=\sigma_1$ in
$\Omega_{2\delta}$, where $B_\delta$ denotes the ball of radius $\delta$, and
$\Omega_{2\delta}=\Rdm \setminus B_{2\delta}$.

The function $\bar \sigma_1$ is smooth and decay at infinity, and then it satisfies the hypothesis of Lemma B.1 in   \cite{Dain:2014iba}.  Hence, the following bounds holds 
\begin{equation}
  \label{eq:4}
  \int_{\Rdm}\left(\left(\Ldt\bar\sigma_1\right)^{2}+|\partial\bar\sigma_{1}|^{2}\right)\rho d\rho dz \geq C\,  \underset{\Rdm}{\sup}|\bar\sigma_1|
\end{equation}
where $C$ is a numerical constant independent of $\bar\sigma_1$ (see also
equations (121) and (122) in \cite{Dain:2014iba} to handle the term with the
Laplacian).

Since  $\sigma_1=\bar \sigma_1 $ in $\Omega_{2\delta}$, we have
\begin{equation}
  \label{eq:5}
  \underset{\Rdm}{\sup}|\bar\sigma_1| \geq \underset{\Omega_{2\delta}}{\sup}|\bar\sigma_1|=\underset{\Omega_{2\delta}}{\sup}|\sigma_1|. 
\end{equation}
That is, if we can bound the integral in the left hand side of the inequality
(\ref{eq:4}) by the energies $m_2$ and $\mt$ then the desired estimate
(\ref{cota3}) follows. To bound this integral we proceed as follows.

We decompose the domain of integration $\Rdm$ in (\ref{eq:4}) in three region
$\Rdm=\Omega_{2\delta}+A_{2\delta} +B_\delta$, where
$A_{2\delta} = B_{2\delta}\setminus B_\delta$.

Define the constant $C_\delta$ by
\begin{equation}
  \label{eq:3}
  C_\delta= \underset{\Omega_{\delta}}{\min} \left\lbrace e^{-2(\sigma_{0}+q_{0})}\right\rbrace.
\end{equation}

For the region $B_\delta$ we have that, by construction,  $\bar\sigma_1=0$ and
hence the integral (\ref{eq:4}) is trivial in  $B_\delta$. 
For the region $\Omega_{2\delta}$ we have $\bar
\sigma_1=\sigma_1$.  For the term with first derivatives we obtain
\begin{align}
  \label{eq:6fd}
 C m_2 &\geq  \int_{\Rdm}|\partial\sigma_{1}|^{2}\,\rho d\rho dz, \\   
&\geq \int_{\Omega_{2\delta}} |\partial\sigma_{1}|^{2} \,\rho d\rho dz,  \label{eq:6fd1}\\
& = \int_{\Omega_{2\delta}} |\partial\bar \sigma_{1}|^{2} \,\rho d\rho dz. \label{eq:6fd12}
\end{align}
Where in (\ref{eq:6fd}) we have used the bound (\ref{cota1}). For the terms with the Laplacian we have
 \begin{align}
  \label{eq:split} 
 C m_2 &\geq \int_{\Rdm} e^{-2(q_{0}+\sigma_{0})}\left(\Ldt\sigma_1\right)^{2} \,\rho d\rho dz, \\   
 &\geq \int_{\Omega_{2\delta}} e^{-2(q_{0}+\sigma_{0})}\left(\Ldt\sigma_1\right)^{2} \,\rho d\rho dz, \\  
&\geq C_\delta  \int_{\Omega_{2\delta}} \left(\Ldt\sigma_1\right)^{2} \,\rho
  d\rho dz, \label{eq:split1} \\
&= C_\delta  \int_{\Omega_{2\delta}} \left(\Ldt\bar\sigma_1\right)^{2} \,\rho d\rho dz,
\end{align}
where in line (\ref{eq:split1}) we have used the definition (\ref{eq:3}). 

It remains only to bound the integral in the transition region $A_{2\delta}$.  For the first derivatives we have
\begin{equation}
  \label{eq:1}
  \partial \bar \sigma_1= (1-\chi_\delta)\partial \sigma_1 -\sigma_1 \partial \chi_\delta, 
\end{equation}
and then we obtain the pointwise estimate
\begin{align}
  \label{eq:6}
    |\partial \bar \sigma_1|^2 & \leq 2 (1-\chi_\delta)^2 |\partial \sigma_1|^2 +2 \sigma^2_1 |\partial \chi_\delta|^2,\\
& \leq  C_\delta \left(|\partial \sigma_1|^2 + \sigma^2_1  \right), 
\end{align}
where we have used that the derivatives of $\chi_\delta$ are bounded by a
constant $C_\delta$ that depends only on $\delta$.  Integrating (\ref{eq:6}) on
$A_{2\delta}$ and using the bound (\ref{cota1}) we obtain
\begin{equation}
  \label{eq:7}
  C_\delta m_2 \geq  \int_{A_{2\delta}} |\partial\bar \sigma_{1}|^{2} \,\rho d\rho dz.
\end{equation}
For the term with the Laplacian we proceed in a similar way, we have
\begin{align}
  \label{eq:8}
  \Ldt\bar\sigma_1=  (1-\chi_\delta)  \Ldt\sigma_1-2\partial \sigma_1 \partial \chi_\delta - \sigma_1 \Ldt \chi_\delta.
\end{align}
Then we obtain
\begin{equation}
  \label{eq:9}
   \left( \Ldt\bar\sigma_1\right)^2 \leq  C_\delta \left(\left( \Ldt \sigma_1\right)^2+ |\partial \sigma_1|^2 + \sigma^2_1   \right),
\end{equation}
where we have used again that all derivatives of $\chi_\delta$ are bounded by a
constant $C_\delta$ . Integrating (\ref{eq:9}) on $A_{2\delta}$, using the
definition (\ref{eq:3}) and the bound we finally obtain
\begin{equation}
  \label{eq:10}
  C_\delta  \int_{A_{2\delta}} \left(\Ldt\bar\sigma_1\right)^{2} \,\rho d\rho dz\leq m_2,
\end{equation}
and hence, collecting all the bounds, we have proved
\begin{equation}
  \label{eq:11}
C_\delta  \int_{\Rdm}\left(\left(\Ldt\bar\sigma_1\right)^{2}+|\partial\bar\sigma_{1}|^{2}\right)\rho d\rho dz \leq m_2.  
\end{equation}

To obtain the bound (\ref{cota4}) for $\bar \omega_1$ the argument is similar.

\section*{Acknowledgments}
It a pleasure to thank Martin Reiris for discussions.  This work was supported
in by grant PICT-2010-1387 of CONICET (Argentina) and grant Secyt-UNC
(Argentina).

\appendix

\section{Axially symmetric linear perturbations of the extreme Kerr black hole}
\label{linearized-equations}
In this appendix we summarize the some of the equations obtained in
\cite{Dain:2014iba} for axially symmetric perturbations for the extreme Kerr
black hole in the maximal-isothermal gauge. In this gauge, Einstein equations
are naturally divided into three groups: evolution equations, constraint
equations and gauge equations.  The evolution equations are further divided
into two groups, evolution equations for the dynamical degree of freedom
$(\sigma_1, \omega_1)$ and evolution equations for the metric (which is
determined by the function $q_1$) and second fundamental form $\chi_{AB}$. In
the proof of theorem \ref{theorem1} we use the evolution equations for
$(\sigma_1, \omega_1)$ given by
\begin{align}
 \label{eq:evolucion-sigma-kerr}
 -e^{2(\sigma_{0}+q_{0})}\dot{p}+^{(3)}\Delta\sigma_{1} &=\dfrac{2}{\eta^2_0}\left(\sigma_{1}|\partial\omega_{0}|^{2}-\partial\omega_{1}\partial\omega_{0} \right),\\
 \label{eq:evolucion-omega-kerr}
 -e^{2(\sigma_{0}+q_{0})}\dot{\po}+^{(3)}\Delta\omega_{1} &=4\dfrac{\partial_{\rho}\omega_{1}}{\rho}+
 2\partial\omega_{1}\partial\sigma_{0}+2\partial\omega_{0}\partial\sigma_{1},
\end{align}
with
\begin{align}
 \label{eq:pkerr}
 p &=\dot{\sigma}_{1}-2\dfrac{\beta_{1}^{\rho}}{\rho}-\beta_{1}^{A}\partial_{A}\sigma_{0},\\
 \label{eq:dkerr}
 \po &=\dot{\omega}_{1}-\beta_{1}^{A}\partial_{A}\omega_{0}.
\end{align}
In these equations the indices $A,B\cdots$ are 2-dimensional, they have the
values $\rho, z$ and $\beta_{1}^{A}$ represents the shift vector of the
foliation. The crucial property of
these equations is that there exists a conserved mass given by
\begin{equation}
 \label{eq:3mass}
 m_{2}=\dfrac{1}{16}\int_{\Rdm}\varepsilon_{2} \, d\rho dz,
\end{equation}
where the positive definite  energy density $\den_2$ is given by 
\begin{multline}
    \label{eq:2epsilon} 
\frac{\den_2}{\rho}=2e^{2(\sigma_{0}+q_{0})}p^{2}+ 2\dfrac{e^{2(\sigma_{0}+q_{0})}}{\eta_{0}^{2}}\po^{2} 
 +4e^{-2u_{0}} \chi^{AB}_{1}\chi_{1AB} + \\
+\left(\partial\sigma_{1}+\omega_{1}\eta^{-2}_{0}\partial\omega_{0}\right)^{2} +\left(\partial\left(\omega_{1}\eta_{0}^{-1}\right)-\eta_{0}^{-1}\sigma_{1}\partial\omega_{0}\right)^{2}
 +\left(\eta_{0}^{-1}\sigma_{1}\partial\omega_{0}-\omega_{1}\eta_{0}^{-2}\partial\eta_{0}\right)^{2}.
  \end{multline}
  See \cite{Dain:2014iba} for the proof.  For completeness, we have written the
  explicit expressions (\ref{eq:pkerr})--(\ref{eq:dkerr}) for the functions $p$
  and $d$ which involve the shift vector. However, we not make use of them in
  this article. The important point is that the functions $p$ and $d$ appear in
  the energy density (\ref{eq:2epsilon}).

The higher  order mass is given by
\begin{equation}
 \label{eq:3mass2}
 \mt=\dfrac{1}{16}\int_{\Rdm}\bar{\varepsilon}_{2}\, d\rho dz,
\end{equation}
with  energy density $\bar{\varepsilon}_2$ is given by 
\begin{multline}
    \label{eq:2epsilon2} 
\frac{\bar{\varepsilon}_2}{\rho}=2e^{2(\sigma_{0}+q_{0})}\dot{p}^{2}+ 2\dfrac{e^{2(\sigma_{0}+q_{0})}}{\eta_{0}^{2}}\dot{\po}^{2} 
 +4e^{-2u_{0}} \dot{\chi}^{AB}_{1}\dot{\chi}_{1AB} + \\
+\left(\partial\dot{\sigma}_{1}+\dot{\omega}_{1}\eta^{-2}_{0}\partial\omega_{0}\right)^{2} +\left(\partial\left(\dot{\omega}_{1}\eta_{0}^{-1}\right)-\eta_{0}^{-1}\dot{\sigma}_{1}\partial\omega_{0}\right)^{2}
 +\left(\eta_{0}^{-1}\dot{\sigma}_{1}\partial\omega_{0}-\dot{\omega}_{1}\eta_{0}^{-2}\partial\eta_{0}\right)^{2}. 
\end{multline}

\section{Extreme Kerr black hole}
\label{s:falloff}

The extreme Kerr black hole solution depends only on one parameter $m_0$, which
represents the total mass of the black hole. In the maximal-isothermal gauge,
the relevant functions used in this article associated with this solution are:
the square norm and the twist of the axial Killing vector denoted by $\eta_0$
and $\omega_0$ respectively and the function $q_0$ which determines the
intrinsic metric of the $t=constant$ slices of the foliation. The function
$\sigma_0$ is calculated from $\eta_0$ by equation (\ref{eq:app35}).  For the
explicit expression for these functions and further details see Appendix A in
\cite{Dain:2014iba}. In this article we will only use the following properties
of these functions.

They satisfy the stationary equations
\begin{align}
\Ldt
\sigma_{0} &=-\dfrac{\arrowvert\partial\omega_{0}\arrowvert^{2}}{\eta_{0}^{2}}, 
\label{eq:sta-sigma}\\
 \partial^{A}\left(\dfrac{\rho\partial_{A}\omega_{0}}{\eta_{0}^{2}}\right) &
 =0. \label{eq:sta-omega}  
\end{align}
Note that equation (\ref{eq:sta-sigma}) is equivalent to 
\begin{equation}
 \label{eq:7eta}
 \Ldt (\ln\eta_{0})=-\dfrac{|\partial\omega_{0}|^{2}}{\eta_{0}^{2}},
\end{equation}
where we have used equation (\ref{eq:app35})  and
\begin{equation}
  \label{eq:21}
  \Ldt(\ln \rho)=0. 
\end{equation}

They satisfies the following elementary inequalities in $\Rdm$
\begin{align}
  \label{eq:falloffso}
  \frac{|\partial \omega_0|^{2}}{\eta^2_0} &\leq \frac{C}{r^{2}}, \\
\label{eq:falloffso2}
  |\partial \sigma_0|^2 &\leq \frac{C}{r^{2}},\\
\label{eq:combi}
  e^{-2 (\sigma_0+q_0)} \frac{|\partial \omega_0|^{2}}{\eta^2_0}  &\leq C, \\
\label{eq:combi2}
e^{-2(\sigma_{0}+q_{0})}|\partial\sigma_{0}|^{2}&\leq C,
\end{align}
where the positive constant $C$ depends only on $m_0$.


\end{document}